\documentclass[twoside]{deon16}

\usepackage{deon16macro}

\usepackage{graphicx}
\usepackage{amsmath}
\usepackage{amssymb}

\usepackage{amsmath,stmaryrd}
\usepackage[all]{xy}
\usepackage{pgf}
\usepackage{tcolorbox}
\usepackage[utf8]{inputenc}	
\usepackage{enumerate}
\usepackage{rotating}
\usepackage{bussproofs}
\usepackage{tikz}
\usetikzlibrary{trees}
\usepackage{enumitem}
\usepackage{marvosym} 

\usepackage{tikz}
\usetikzlibrary{positioning,arrows,calc,decorations.markings,fit}

\tikzset{
modal/.style={>=stealth',shorten >=1pt,shorten <=1pt,auto,node distance=1.5cm,semithick},
world/.style={circle,draw,minimum size=0.5cm,fill=gray!15},
point/.style={circle,draw,inner sep=0.5mm,fill=black},
reflexive above/.style={->,loop,looseness=7,in=120,out=60},
reflexive below/.style={->,loop,looseness=7,in=240,out=300},
reflexive left/.style={->,loop,looseness=7,in=150,out=210},
reflexive right/.style={->,loop,looseness=7,in=30,out=330}
}

\usepackage{amssymb, upgreek}
\usepackage{mathtools}  
\usepackage{graphicx}
\usepackage{xfrac}
\usepackage[colorinlistoftodos]{todonotes}
\usepackage{color, multicol}
\usepackage{rotating}

\usepackage[scr]{rsfso}

\definecolor{green2}{RGB}{154, 205, 50}

\newcommand {\Diam}{\scalebox{0.9}{\rotatebox[origin=c]{45}{$\Box$}}}

\newcommand{\lang}{\mathcal{L}}

\newcommand{\Oi}{\otimes_{i}}
\newcommand{\lb}{\langle}
\newcommand{\rb}{\rangle}

\newcommand{\agdia}{\lb i \rb}
\newcommand{\agbox}{[i]}

\newcommand{\ODi}{\ominus_i}
\newcommand{\sti}{\mathsf{STIT}}
\newcommand{\opt}{R}

\newcommand{\baselogic}{\mathsf{DS}}
\newcommand{\ds}[1]{\mathsf{DS}_{#1}}
\newcommand{\dsx}[2]{\mathsf{DS}_{#1}{\mathsf{#2}}}
\newcommand{\gds}[2]{\mathsf{G3DS}_{#1}{\mathsf{#2}}}

\newcommand{\oilp}{\mathsf{OiLP}}
\newcommand{\oipp}{\mathsf{OiAP}}
\newcommand{\oia}{\mathsf{OiA}}
\newcommand{\oiv}{\mathsf{OiV}}
\newcommand{\oir}{\mathsf{OiR}}
\newcommand{\oio}{\mathsf{OiO}}
\newcommand{\oiao}{\mathsf{OiA{+}O}}
\newcommand{\oictrl}{\mathsf{OiCtrl}}
\newcommand{\oinc}{\mathsf{OiNC}}
\newcommand{\oina}{\mathsf{OiNA}}

\newcommand{\RT}{\textbf{RT}}

\newcommand{\disrule}{(\lor)}

\newcommand{\dtwo}{(\mathsf{D2}_{\text{$i$}})}
\newcommand{\dthree}{(\mathsf{D3}_{\text{$i$}})}
\newcommand{\dfour}{(\mathsf{D4}_{\text{$i$}})}
\newcommand{\dfivei}{(\mathsf{D5}_{\text{$i$}}^1)}
\newcommand{\dfiveii}{(\mathsf{D5}_{\text{$i$}}^2)}
\newcommand{\dfive}{(\mathsf{D5}_{\text{$i$}})}
\newcommand{\done}{(\mathsf{D1}_{\text{$i$}})}
\newcommand{\RB}{R_{\Box}}
\newcommand{\RI}{R_{[i]}}
\newcommand{\RO}{R_{\Oi}}

\newcommand{\R}{R}

\newcommand{\rel}{\mathcal{R}}

\newcommand{\ioa}{(\mathsf{IOA})}

\newcommand{\id}{(\mathsf{id})}

\newcommand{\Oir}{(\Oi)}
\newcommand{\Odiamrule}{(\ominus_{i})}

\newcommand{\disr}{(\vee)}
\newcommand{\conr}{(\wedge)}
\newcommand{\settr}{(\Box)}
\newcommand{\settdiar}{(\Diamond)}
\newcommand{\stitr}{(\agbox{})}
\newcommand{\stitdiar}{(\agdia)}

\EnableBpAbbreviations

\def\lastname{van Berkel, Lyon}

\begin{document}

\begin{frontmatter}
  \title{The Varieties of Ought-implies-Can and Deontic STIT Logic}
  \author{Kees van Berkel}\footnote{We would like to thank the reviewers of DEON2020 for their useful comments. This work is funded by the projects WWTF MA16-028, FWF I2982 and FWF W1255-N23. For questions and comments please contact \texttt{kees@logic.at}.}
 \address{Institute of Logic and Computation, \\
 Technische Universit\"at Wien, 1040 Wien, Austria 
}
\author{Tim Lyon
}

\address{Institute of Artificial Intelligence, \\
Technische Universit\"at Dresden, 
01069 Dresden,
Germany
}



\begin{abstract}
$\sti$ logic is a prominent framework for the analysis of multi-agent choice-making. In the available deontic extensions of $\sti$, the principle of Ought-implies-Can (OiC) fulfills a central role. However, in the philosophical literature a variety of alternative OiC interpretations have been proposed and discussed. This paper provides a modular framework for deontic $\sti$ that accounts for a multitude of OiC readings. In particular, we discuss, compare, and formalize ten such readings. We provide sound and complete sequent-style calculi for all of the various $\sti$ logics accommodating these OiC principles. We formally analyze the resulting logics and discuss how the different OiC principles are logically related. In particular, we propose an endorsement principle describing which OiC readings logically commit one to other OiC readings.
\end{abstract}
  \begin{keyword}
Deontic logic, $\sti$ logic, Ought implies can, Labelled sequent calculus
  \end{keyword}
 \end{frontmatter}

\section{Introduction}\label{sect:intro}

From its earliest days, the development of deontic logic has been accompanied by the observation that reasoning about duties 
 is essentially connected to \textit{praxeology}, that is, 
 the theory of agency (e.g.~\cite{Cas72,Mey88,Wri68}). A prominent modal framework developed for the analysis of multi-agent interaction and choice 
  is the logic of `Seeing To It That' \cite{BelPerXu01} (henceforth, $\sti$), and its potential for \textit{deontic 
  reasoning} was recognized from the outset~\cite{BelPer88}. Despite several philosophical investigations of the subject \cite{Bar93,
HorBel95}, concern for its formal specification lay dormant until the beginning of this century when a thorough investigation of deontic $\sti$ logic was finally conducted \cite{Hor01,Mur04}. Up to the present day, deontic $\sti$ continues to receive considerable attention, being applied to epistemic \cite{Bro11}, temporal \cite{BerLyo19b}, and juridical contexts \cite{LorSar15}. 

The traditional deontic $\sti$ setting \cite{Hor01} is rooted in a utilitarian approach to agential choice, 
 which enforces certain minimal properties on its agent-dependent obligation operators. In particular, it implies a version of the eminent \textit{Ought-implies-Can} principle (henceforth, OiC), a metaethical principle postulating that `what an agent ought to do, the agent can do'. OiC has a long history within moral philosophy and can be traced back to, for example, Aristotle \cite[VII-3]{Ari00}, 
 or 
  the ``Roman legal maxim \textit{impossibilium nulla obligatio est}’’ \cite{Vra07}. Still, it is often accredited to the renowned philosopher Immanuel Kant \cite[A548/B576]{Kan00}. Aside from debates on whether OiC should be adopted at all~\cite{Gra11,Sak00}, most discussions revolve around which \textit{version} of the principle should be endorsed. Notable 
   positions 
   have been
    taken up by Hintikka \cite{Hin70}, Lemmon \cite{Lem62}, Stocker \cite{Sto71}, Von Wright \cite{Wri63}, and, more recently, Vranas \cite{Vra07}. However, 
     most of these authors advocate readings 
 that are either weaker or stronger than the minimally implied OiC principle of traditional deontic 
   $\sti$. In order to formally investigate these different readings, it is necessary to modify and fine-tune the traditional framework.

The contributions of this work are as follows: First, we discuss, compare, and formalize ten OiC principles occurring in the  philosophical
 literature (Sect.~\ref{sect:principles}). To the best of our knowledge, such a taxonomy 
 of principles has not yet been undertaken (cf.~\cite{Vra07} for an extensive 
  bibliography).  
 The intrinsically agentive setting provided within the $\sti$ paradigm will enable us to conduct a fine-grained analysis of the various renditions of OiC. Still, the available utilitarian characterization of deontic $\sti$ makes it cumbersome to accommodate this multiplicity of principles. For that reason, the present endeavour will take a more modular approach to $\sti$, adopting relational semantics \cite{CiuLor17} through which the use of utilities may be omitted \cite{BerLyo19b}  (Sect.~\ref{sect:logics}).
  
Second, we provide sound and complete sequent-style calculi for all classes 
  of deontic $\sti$ logics accommodating 
 the various kinds of formalized OiC principles 
   (Sect.~\ref{sect:calculi}). In particular, we adopt 
 labelled sequent calculi which 
     explicitly incorporate useful
    semantic information into their rules 
     \cite{Neg14,Vig00}. 
 A general benefit of using sequent-style calculi \cite{Tak13}, in contrast to 
    axiomatic systems, 
    is that the former are 
    suitable for 
    applications (e.g. proof-search and counter-model extraction)~\cite{LyoBer19}. Although this work is not the first to address $\sti$ through alternative proof-systems~\cite{ArkBriBel05,LyoBer19,Wan06}, it is the first to address 
     both the traditional deontic setting ~\cite{Hor01} and a large class of novel deontic $\sti$ logics. 

Last, we will use the resulting deontic $\sti$ calculi to obtain a formal taxonomy of the OiC readings discussed.  The benefit of employing proof theory 
 is twofold: First, we classify the ten OiC principles according to the respective strength of the underlying $\sti$ logics in which they are embedded (Sect.~\ref{sect:analysis}). The calculi can be used to determine which logics subsume each other, giving rise to what we call an \textit{endorsement principle}; it demonstrates which endorsement of which OiC readings logically commits one to endorsing other OiC readings (from the vantage of $\sti$). Second, the calculi can be applied to show the mutual independence of certain OiC readings through the construction of counter-models from failed proof-search. 
  This work will lay 
   the 
 foundations for an extensive 
  investigation of 
  OiC within the realm of agential choice, 
   and 
  future research directions will be addressed in Sect.~\ref{sect:conclusion}.

\section{A Variety of Ought-implies-Can Principles}\label{sect:principles}

The fields of moral philosophy and deontic logic have given rise to a variety of metaethical principles, such as 
``no vacuous obligations’’~\cite{Wri51}, ``
deontic contingency’’~\cite{AndMoo57}, ``
deontic consistency’’~\cite{HilMcN13}, and the principle of ``alternate possibilities''~\cite{Cop17}. 
  One of the most prevalent is perhaps the principle of \textit{``Ought-implies-Can''}. In fact, we will see that each of the former metaethical canons is significant relative to different interpretations of OiC. 
 In this 
     section, we 
      introduce and discuss ten such interpretations of OiC and indicate their relation to the aforementioned metaethical principles. Many philosophers have addressed OiC, and while earlier thinkers (e.g. Aristotle and Kant) only discussed it implicitly, it was made an explicit subject of investigation in the past century. We will focus solely on frequently recurring readings from authors that are---in our opinion---central to the debate. Despite the apparent relationships between some of the considered OiC readings, a precise taxonomy of their logical interdependencies can only be achieved through a formal investigation of their corresponding 
       logics. We will provide such a taxonomy 
       in Sect.~\ref{sect:analysis}.

One of the allures 
of OiC is that it releases agents from alleged duties which are impossible, strenuous, or over-demanding~\cite{Dah74,McC89}. Namely, in its basic 
  formulation---
`what an agent ought to do, the agent can do'---the principle ensures 
 that an agent can only be normatively bound by what it can do, i.e., 
  `what the agent can't do, the agent is not obliged to do'. 
Most disagreement concerning OiC can be understood in terms of the degree to which an agent must be burdened or relieved. In essence, such discussions revolve around the appropriate interpretation of the terms `ought’, `implies’, and predominantly, `can’. 
 In what follows, we take `ought’ to represent agent-dependent obligations and take `implies’ to stand for material implication	 
 (for a discussion see~\cite{AckKuh15,Vra07}). With respect to the term `can', we roughly identify four readings: 
  (i) 
   possibility, (ii) 
  ability, (iii) violability, and (iv) control.  These four concepts give rise to 
  eight OiC principles. 
 We close the section with a discussion of two additional OiC principles which 
  adopt a normative reading of the term `can'.

Throughout our discussion we introduce logical formalizations of the proposed OiC readings that will be made formally precise in subsequent sections. Therefore, it will be useful at this stage to introduce some notation employed in our formal language: we let $\phi$ stand for an arbitrary $\sti$ formula. The connectives $\lnot,\land,$ and $\rightarrow$ are respectively interpreted as `not', `and', and `implies'. Let $[i]$ be the basic $\sti$ operator such that, in the spirit of \cite{BelPerXu01}, we interpret $[i]\phi$ both as `agent $i$ sees to it that $\phi$' and `agent $i$ chooses to ensure $\phi$'. We use the operator $\Box$ to refer to what is `settled true', such that $\Box\phi$ can be read as `currently, $\phi$ is settled true'. 
 The main use of $\Box$ is to discern between those state-of-affairs that can become true---i.e. actual---through an agent's choice 
 and those state-of-affairs that are true---i.e. actual---independent of the agent's choice. For this reason we will also interchangeably employ the term `actual' in referring to $\Box$ 
 (for an extensive discussion see 
  \cite{BelPerXu01}). We take $\Diam$ to be the dual of $\Box$, denoting that some state of affairs is actualizable, i.e., can become actual. Last, we read 
    $\Oi$ as `it ought to be the case for agent $i$ that'.\footnote{We stress that OiC is essentially agentive, but not necessarily referring to choice in particular. For this reason,
     we distinguish `it ought to be the case for agent $i$ that' from the stronger `agent $i$ ought to see to it that'. The latter reading corresponds to the notion of `\textit{dominance ought}' advocated by Horty \cite{Hor01}. Initially, the distinction will be observed for OiC. In Sect.~\ref{sect:analysis} we show how the logics can be expanded to obtain the stronger reading 
     proposed in \cite{Hor01}.} 

\begin{itemize}

\item[$\mathsf{1.}$] 
 \textit{Ought implies Logical Possibility}: $\Oi \phi \rightarrow \lnot\! \Oi\!\lnot \phi$ 
  ($\mathsf{OiLP}$). 
What is obligatory for an agent, 
 should 
  be 
   consistent from an ideal point of view. 

\end{itemize}
 The first principle, which is one of the weakest interpretations of OiC, requires the content of an agent's obligations to be non-contradictory. 
 Within the philosophical literature this interpretation has been referred to 
 as ``ought implies logical possibility’’~\cite{Vra07} and the principle has been generally 
 equated with 
the metaethical principle of ``deontic consistency’’ (e.g.~\cite{Eck82,Lem62}).\footnote{In \cite{Wri81}, Von Wright baptizes $\oilp$ `Bentham's Law' and points out that the canon was already adopted by Mally in what is known as the first attempt to construct a deontic logic.} 
As a 
  minimal constraint on deontic reasoning, the principle is a cornerstone of (standard) Deontic Logic \cite{AndMoo57,HilMcN13,Wri51}, though it has been repudiated by some~\cite{Lem62}.
\begin{itemize}
\item[$\mathsf{2.}$] 
\textit{Ought implies Actually Possible}: 
 $\Oi \phi \rightarrow \Diam \phi$ ($\oipp$). What is obligatory for an agent, 
 should be actualizable. 
\end{itemize}
The above principle is slightly stronger than the previous one: it rules out those conceptual consistencies that might not be realizable at the current moment.
\footnote{In \cite{HilMcN13}, OiC is named `Kant's law' and $\oilp$ and $\oipp$ are classified as weak versions of the law. 
  However, it is open to debate which reading of OiC 
  Kant 
  would 
  admit 
   (e.g.~\cite{Koh15,Tim13}).} That is, the principle requires that norm systems can only demand what can presently become \textit{actual}. For example, 
    `although it is logically possible 
 to open the window, it is currently not actualizable, 
  since I am tied to the chair'.

However, both $\oilp$ and $\oipp$ are arguably too weak, and do not involve the concerned agent 
 whilst interpreting `can'. For instance, although `a moon eclipse' is both logically and actually possible, it should not be considered as something an agent ought to bring about. For this reason, most renditions of OiC involve the agent explicitly:  
\begin{itemize}
\item[$\mathsf{3.}$] 
 \textit{Ought implies Ability}: $\Oi \phi \rightarrow \Diam [i] \phi$ ($\mathsf{OiA}$). What is obligatory for an agent, the agent must have the ability to see to, i.e. the choice to realize.  
\end{itemize}
The above reading 
enforces an explicitly agentive precondition on obligations: it requires ability as the agent's capacity to guarantee the realization of that which is prescribed.\footnote{Similarly, Von Wright distinguishes between human and physical possibility (cf. $\oia$ and $\oipp$, resp.), both implying logical possibility (cf. $\oilp$) as a necessary condition \cite[p.50] {Wri68}.} The concept of ability has many formulations (cf. \cite{Bro11,Bro88,Gol70,Wri63}); for example, it may denote general ability, present ability, potential ability, learnability, know-how, and even technical skill (also, see \cite{McC89,Sto71,Vra07} on the corresponding notion of `inability'). In what follows, we take `ability' to mean a \textit{moment-dependent} possibility for an agent to guarantee that which is commanded through an available \textit{choice}.

Observe that $\oia$ 
 is the principle implied by the traditional, utilitarian based 
  deontic 
   $\sti$ logic \cite{Hor01,Mur04}. However, this OiC reading 
    does not 
    completely capture the notion of `ability’ as 
    generally encountered in the philosophical literature. That is, $\oia$ merely requires that what is prescribed for the agent can be guaranteed through one of the agent's choices, but does not exclude what is called vacuously satisfied obligations. Agents could still have obligations (and corresponding `abilities') to bring about inevitable states-of-affairs, such as the obligation to realize a tautology (cf.~\cite{BerLyo19b}). Philosophical notions of ability regularly ban such consequences by strengthening the concept of ability with either (i) the \textit{possibility} that the obligation may be \textit{violated}, (ii) the agent’s \textit{ability to violate} what is demanded (i.e. an agent may refrain from fulfilling a duty), (iii) the right \textit{opportunity} for the agent to exercise its ability, or (iv) the agent’s \textit{control} over the situation (i.e. the agent's power to decide over the fate of what is prescribed). All of the above conceptions of agency are contingent 
    in nature, that is, they range over state-of-affairs which are capable of being otherwise \cite{HorBel95}. Each notion will be addressed in turn. 
       
\begin{itemize}
\item[$\mathsf{4.}$] 
\textit{Ought implies Violability}: $\Oi\phi\rightarrow \Diam \lnot \phi$  ($\mathsf{OiV}$). 
An agent’s obligation 
must be violable, that is, the opposite of what is prescribed must be possible. 
\end{itemize}

The above principle corresponds to 
 the metaethical principle of 
``no vacuous obligations’’, which ensures that neither tautologies are 
obligatory nor 
contradictions are prohibited~\cite{AndMoo57,HilMcN13,Wri63}. However, in $\oiv$ a violation might still arise through causes external to the agent concerned; 
  e.g. `the prescribed opening of a window, might be closed through a strong gust of wind'.\footnote{Already in \cite
{Wri51} Von Wright posed the `no vacuous obligations' principle as a central principle of deontic logic. 
 There, he referred to it as ``the principle 
 of contingency'', however, 
 contingency requires that an obligation is not only 
 violable, but also satisfiable (cf. $\oio$).}  
 The following principle strengthens this notion by making violability an agentive matter: 
  \begin{itemize}
\item[$\mathsf{5.}$] 
 \textit{Ought implies Refrainability}
 : $\Oi \phi \rightarrow \Diam [i] \lnot [i] \phi$ ($\mathsf{OiR}$). An agent’s obligation 
  must be deliberately violable by the agent, that is, the agent must be able to refrain from satisfying its obligation. 
\end{itemize}
In the jargon of $\sti$, we say that \textit{refraining} from fulfilling one's duty requires ``an embedding of a non-acting within an acting''~\cite[Ch.2]{BelPerXu01}. That is, it requires the possibility to `see to it that one does not see to it that'. However, the two violation principles above are insubstantial when that which is obliged is not possible in the first place.\footnote{We conjecture that this is why Vranas states that $\oir$ is strictly not an OiC principle \cite{Vra07}.} For instance, it is not difficult for an agent to violate an obligation to `create a moon eclipse' (it could not be done otherwise).\footnote{Observe that violability relates strongly to the metaethical principle of ``alternate possibility'', stating that an agent is morally culpable if it could have done otherwise (e.g. \cite{Cop17,Yaf99}).} To avoid such cases, we often find that the ideas from $\mathsf{1{-}5}$ are combined:
\begin{itemize}
\item[$\mathsf{6.}$] 
 \textit{Ought implies Opportunity} ($\mathsf{OiO}$): $\Oi \phi \rightarrow (\Diam \phi \land \Diam \lnot \phi)$. What is obligatory for an agent, must be a contingent state-of-affairs. 
\end{itemize}
The above uses 
 the terms `opportunity' and `contingency' intentionally in an interchangeable manner. 
  Like previous terms, these terms 
 know a variety of readings in the literature 
  (cf. \cite{Cop17,Dah74,Vra07,Wri51}). Nevertheless, what these readings share in relation to OiC is that they refer to the propriety of the circumstances in which the agent is required to fulfill its duty. Minimally, opportunity and contingency both require that a state-of-affairs within the scope of an active norm must be presently manipulable; i.e. the state-of-affairs can still become true or false.\footnote{A more fine-grained distinction can be made: in temporal settings a state-of-affairs can be occasionally true and false (i.e. contingent), despite the fact that at the present moment it is settled true and thus beyond the scope of the agent's influence (i.e. there is no opportunity). In the current atemporal $\sti$ setting, this will not be explored.} This interpretation of $\oio$ is related
 to what Von Wright has in mind when he talks about the opportunity 
 to interfere with the course of nature \cite{Wri63
}, 
and to 
 Anderson and Moore's claim that sanctions (i.e. violations) must be both provokable and 
 avoidable, viz. contingent~\cite{AndMoo57}.




Taking the above a step further,  
 agency can be more precisely described as the agent's \textit{ability} together with the right \textit{opportunity}. 
 Following Vranas \cite{Vra07}, the latter component 
 specifies ``the situation hosting the event in which the agent has to exercise her ability''. 
 The following principle merges these ideas:  

\begin{itemize}
\item[$\mathsf{7.}$] 
 \textit{Ought implies Ability and Opportunity}: $\Oi \phi \rightarrow (\Diam [i] \phi \land \Diam \phi \land \Diam \lnot \phi)$ ($\mathsf{OiA+O}$). What is obligatory for an agent, must be a contingent state-of-affairs whose truth the agent has the ability to secure.\footnote{In basic atemporal $\sti$ the occurrence of $\Diam\phi$ in the consequent of $\oiao$ can be omitted since it is strictly implied by $\Diam[i]\phi$; that is, if $\phi$ can be the result of an agent's choice, then by definition it can be actualized. For the sake of completion we leave $\Diam\phi$ present in $\oiao$.} 
\end{itemize}
 The above 
  is the first completely agentive OiC principle, 
 making that which is obligatory fall, in all its facets, within the reach of the agent. Such a reading of OiC can be said to be truly deliberative and both Vranas \cite{Vra07} and Von Wright \cite{Wri63} appear to endorse a principle similar to $\oiao$. However, there is an even stronger reading which restricts norms to those state-of-affairs within the agent's complete \textit{control}:
\begin{itemize}
\item[$\mathsf{8.}$] 
 \textit{Ought implies Control}: $\Oi\phi \rightarrow (\Diam [i] \phi \land \Diam [i] \lnot \phi)$ ($\mathsf{OiCtrl}$). What is obligatory for an agent, 
  the agent must have the ability to see to and the agent must have the ability to see to it that the obligation is violated. 
\end{itemize}
This reading, arguably advocated by Stocker \cite{Sto71}, requires that an agent can act \textit{freely}: ``it has often been maintained that we act freely in doing or not doing an act only if we both can do it and are able not to do it'' \cite
{Sto71}.\footnote{In the above quote, `able not to do [$\phi$]' can also be interpreted as $\Diam [i] \lnot [i] \lnot \phi$, instead of $\Diam [i]\lnot\phi$. The resulting principle would then 
 equate 
  with the weaker $\oiao$ in basic $\sti$.} This last, perhaps too strong, instance of OiC implies that an agent is only subject to norms whose subject matter is within the \textit{power} of the agent.

\begin{figure}
\def\arraystretch{1.0}{
{\small
\begin{tabular}{l l l l} 
 $\!\!$Label  & Ought implies... & Formalized &  References \\
\hline
 $\!\!\mathsf{OiLP} $  & 
 Logical Possibility $\!\!\!$ & $\Oi \phi \rightarrow \lnot \!\Oi\!\lnot \phi$ &   \cite{AndMoo57}, \cite{Eck82}, \cite{Wri51}, \cite{Wri81} \\
$\!\!\oipp $ & 
Actually Possible & $\Oi \phi \rightarrow \Diam \phi$ &  \cite{Eck82}, \cite[Ch.3]{Hor01} \\
$\!\!\mathsf{OiA} $& 
Ability & $\Oi \phi \rightarrow \Diam [i] \phi$ &  \cite[Ch.4]{Hor01}, \cite[Ch.7]{Wri63} \\
$\!\!\mathsf{OiV} $& 
Violability & $\Oi \phi \rightarrow \Diam \lnot \phi$ &  \cite{AndMoo57}, \cite{Dah74}, \cite{Gol70}, \cite[Ch.8]{Wri63}$\!\!\!\!$ \\
$\!\!\mathsf{OiR} $ & 
Refrainability & $\Oi \phi \rightarrow \Diam [i]\lnot [i] \phi$ & \cite{Gol70} \\
$\!\!\mathsf{OiO} $ & 
Opportunity & $\Oi \phi \rightarrow (\Diam \phi \land \Diam \lnot \phi)$ & \cite{AndMoo57}, \cite{Cop17}, \cite{Dah74}, \cite{Wri51}, \cite{Wri68} \\
$\!\!\mathsf{OiA{+}O}\!\! $& 
Ability and Opp. 
& $\Oi \phi \rightarrow (\Diam [i] \phi \land \Diam\phi \land \Diam \lnot \phi)\!\!\!$ & \cite{AckKuh15}, \cite{Koh15}, \cite{Vra07}, \cite{Wri63} \\
 $\!\!\mathsf{OiCtrl} $& 
 Control & $\Oi \phi \rightarrow (\Diam [i] \phi \land \Diam [i] \lnot \phi)\!\!$ &\cite{Dah74}, \cite{Sto71}, \cite{McC89} \\
 $\!\!\oinc$ & Normatively Can & $\Oi\phi \rightarrow \Oi\Diam\phi$ &  \cite{AckKuh15}, \cite{Hin70}\\
 $\!\!\oina$ & Normatively Able & $\Oi\phi \rightarrow \Oi \Diam [i] \phi$ & \cite{AckKuh15}, \cite{Hin70} \\
\end{tabular}

\caption{List of the ten OiC principles together with their 
treatment in the literature.}\label{fig:oic}
}
}
\end{figure}


In all its 
  readings, OiC has still been regarded as too strong. 
   For example, Lemmon 
    challenged the legitimacy of $\oilp$ in light of the existence of moral dilemmas \cite{Lem62}. Other philosophers, like Hintikka \cite{Hin70}, adopted more modest standpoints toward OiC, suggesting 
weaker, normative versions of the principle. In light of the latter, it has been argued that OiC is 
 dispositional, merely capturing a 
 normative attitude towards OiC \cite{AckKuh15}. 
  Two approaches present themselves: (i) `it \textit{ought to be} the case that what morality prescribes is possible' 
  or (ii) `
   it \textit{ought to be possible} for an agent to fulfill its obligations'.\footnote{Hintikka advocates the first possibility; i.e. ``$\mathcal{O}(\mathcal{O}\phi\rightarrow \Diam\phi)$'' \cite{Hin70}. However, one could argue that 
  the first occurrence of $\mathcal{O}$ is actually agent-\textit{in}dependent, 
  and the latter agent-dependent. 
  } The former does not correspond to an OiC principle, but only expresses that OiC \textit{should} hold as a metaethical principle (we return to this in Sect.~\ref{sect:analysis}). The latter approach does provide 
  OiC principles---we consider two possible readings: 
 
 \begin{itemize}
\item[$\mathsf{9.}$] 
 \textit{Ought implies Normatively Can}: $\Oi\phi \rightarrow \Oi \Diam \phi$ ($\mathsf{OiNC}$). What 
 is obligatory for an agent, 
 ought to be actually possible (for the agent). 
\end{itemize}

 \begin{itemize}
\item[$\mathsf{10.}$] 
 \textit{Ought implies Normatively Able}: $\Oi\phi \rightarrow \Oi \Diam [i] \phi$ ($\mathsf{OiNA}$). What is obligatory for an agent, ought to be actualizable through the agent's behavior.
\end{itemize}
Hence, both $\oinc$ and $\oina$ require that, `if $\phi$ ought to be the case for agent $i$, it ought to be the case for agent $i$ that $\phi$ is actually possible (as a result of the agent's choice-making)'. In Fig.~\ref{fig:oic}, the ten principles are collected and associated with references to the various authors that treat such principles.
 
It is not our aim to decide which OiC principle should be adopted, as good cases have been made for each. Instead, our present aim is as follows: first, we appropriate the framework of $\sti$ such that all ten principles can be explicitly formulated (Sect.~\ref{sect:calculi}). Second, we use the resulting logics to formally determine the logical relations between the ten principles (Sect.~\ref{sect:analysis}). The final result will be a logical hierarchy of OiC principles, identifying which principles subsume others and which are mutually independent within the setting of $\sti$.

\section{Deontic STIT Logic for Ought-implies-Can}\label{sect:logics}

In this section, 
 we will introduce a general deontic $\sti$ language and semantics whose modularity enables us to define a collection of deontic $\sti$ logics 
  that will accommodate the variety of OiC principles discussed previously. It will suffice to consider a multi-agent modal 
   language containing the basic $\sti$ operator (i.e. the Chellas $\sti$) and the `settled true' operator, extended with agent-dependent deontic operators. 

\begin{definition}\textbf{(The Language $\lang_{n}$)}\label{def:language} Let $Ag = \{1,2,...,n\}$ be a finite set of agent labels and let $Atm =\{p_1,p_2,p_3...\}$ be a denumerable set of propositional atoms. The language $\lang_n$ is defined via the following BNF grammar:
$$\phi ::= p \ | \ \neg p \ | \ \phi \lor \phi \ | \ \phi \land \phi \ | \ \Box \phi \ | \ \Diam \phi \ | \ [i] \phi \ | \ \agdia \phi \ | \ \Oi\phi \ | \ \ODi \phi$$
where $i\in Ag$ and $p \in Atm$.

\end{definition}

We note that the formulae of $\lang_n$ are defined in negation normal form. In line with~\cite{BerLyo19a,LyoBer19}, we opt for this notation because it will 
 enhance the readability of the technical part of this paper. Namely, negation normal form will reduce the number of logical rules needed in our sequent-style calculi (see Sect.~\ref{sect:calculi}), and will simplify the structure of sequents used in derivations (see Sect.~\ref{sect:analysis}). Briefly, the negation of a formula $\phi \in \lang_{n}$, denoted by $\neg \phi$, can be obtained by replacing each positive propositional atom $p$ with its negation $\neg p$ (and vice versa), each $\land$ with $\lor$ (and vice versa), and each modal operator with its corresponding dual (and vice versa).
   
The logical connectives $\lor$ and $\land$ stand for `or' and `and', respectively. Other connectives and abbreviations are defined accordingly: $\phi \rightarrow \psi \textit{ iff } \neg \phi \lor \psi$, $\phi \equiv \psi \textit{ iff } (\phi \rightarrow \psi) \land (\psi\rightarrow \phi)$, $\top \textit{ iff } p \lor \neg p$, and 
    $\bot \textit{ iff } p \land  \neg p$. The modal operators $\Box$, $[i]$, and $\Oi$ express, respectively, `currently, it is settled true that', 
  `agent $i$ sees to it that', and `it ought to be the case for agent $i$ that'. We take $\Diam$, $\langle i \rangle$, and $\ODi$ as their respective duals. Last, we interpret $\ODi$ as `it is not obligatory for agent $i$ that not’ (a similar interpretation is applied to $\Diam$ and $\langle i\rangle$). (NB. negation normal form requires us to take diamond-modalities as primitive.)\footnote{In line with \cite{Mur04}, we take the concatenation $\Oi[i]$ to stand for `agent $i$ ought to see to it that',  thus expressing 
   the stronger agentive reading of obligation 
    defended by \cite{Hor01} (also, see footnote 2). However, whether 
    $\Oi[i]$ will capture the intended logical behavior of this 
    reading
     will depend on the adopted class of $\sti$-frames. We will discuss this in 
      Sect.~
\ref{sect:analysis}.} 

\subsection{Minimal Deontic $\sti$ Frames}

Since we are dealing with an atemporal $\sti$ language, we can forgo the traditional semantics of branching time frames with agential choice functions ~\cite{BelPerXu01}. Instead, we adopt a more modular approach using relational semantics \cite{CiuLor17}. As shown in \cite{HerSch08}, it suffices to semantically characterize basic $\sti$ using frames that only model moments partitioned into equivalence classes, with the latter representing the choices available to the agents at the respective moment. As our starting point, we propose the following minimal deontic $\sti$ models:


\begin{definition}\textbf{(Frames and Models for $\baselogic_{n}$})\label{def:frames-models} A \emph{$\ds{n}$-frame} is defined to be a tuple $F = \langle W, \R_{\Box}, \{\R_{[i]} \ | \ i \in Ag\}, \{ \opt_{\Oi} \ | \ i \in Ag\} \rangle$ with $n=|Ag|$. Let $\R_{\alpha} \subseteq W \times W$ and $\R_{\alpha}(w) := \{v\in W \ | \ (w,v) \in R_{\alpha}\} $ for $\alpha \in \{\Box \} \cup \{[i], \Oi \ | \ i\in Ag\}$. Let $W$ be a non-empty set of worlds $w,v,u...$ where:

\vspace{5pt}
\renewcommand{\arraystretch}{1.1}
\noindent \begin{tabular}{p{1em} p{30pt} p{295pt}}
 & {\rm \textbf{C1}} & $\R_{\Box}$ is an equivalence relation. \\
 & {\rm \textbf{C2}} & For all $i\in Ag$, $\R_{[i]} \subseteq \RB
 $ is an equivalence relation. \\
 & {\rm \textbf{C3}} & For all $w\in W$ and 
  all $u_{1},...,u_{n} \in R_{\Box}(w)$, $\bigcap_{i \in Ag} \R_{[i]}(u_{i}) \neq \emptyset$. \\
 & {\rm \textbf{D1}} & For all $w,v, u \in W$, if $R_{\Box}wv$ and $\opt_{\Oi}wu$, then $\opt_{\Oi}vu$. \\
\end{tabular}
\vspace{5pt}

\noindent A \emph{$\baselogic_{n}$-model} is a tuple $M = (F,V)$ where $F$ is a $\ds{n}$-frame and $V$ is a valuation function mapping  propositional
 atoms to subsets of $W$, i.e. $V{:}\ Atm \mapsto \mathcal{P}(W)$.
\end{definition}

In Def.~\ref{def:frames-models}, property \textbf{C1} stipulates that $\ds{n}$-frames are partitioned into $R_{\Box}$-equivalence classes, 
which we will refer to as \textit{moments}. Intuitively, a moment is a collection of worlds that can become actual. For every agent in the language, \textbf{C2} partitions moments into equivalence classes, representing the agent's \textit{choices} 
  at such moments. The elements of a choice represent those worlds that can become actual through exercising that choice. \textbf{C3} captures the pivotal $\sti$ principle called `independence of agents', ensuring that all agents can jointly perform their available choices; i.e. simultaneous choices are consistent (cf.~\cite{BelPerXu01}). \textbf{D1} enforces that ideal worlds do not vary from different perspectives within a single moment; i.e. an ideal world is ideal from the perspective of the entire moment. In addition, \textbf{D1} states that obligations are moment-dependent; i.e. obligations 
   might vary from moment to moment. We emphasize that the class of $\ds{n}$-frames does not require that worlds ideal at a certain moment lie within that very moment. Hence, 
    what is ideal might not be realizable by any of the agents' (combined) choices, and so, might be beyond the grasp of agency.\footnote{Traditional deontic $\sti$ confines ideal worlds to moments since it restricts the evaluation of utilities to moments \cite{Hor01}. Consequently, $(\Oi \phi \rightarrow \ODi \phi) \equiv (\Oi \phi \rightarrow \Diam \phi)$ is valid for the traditional approach, and thus, logical and actual possibility coincide. 
    Our alternative semantics enables us to differentiate between $\oilp$, $\oipp$ and a variety of other OiC principles.}
 
\begin{definition}\textbf{(Semantics for $\lang_{n}$)}\label{def:semantics} Let $M$ be a $\baselogic_{n}$-model and let $w\in W$ of $M$. The \emph{satisfaction} of a formula $\phi\in \lang_{n}$ in $M$ at $w$ is defined accordingly:
 
\vspace{-6pt}
\setlength{\columnsep}{-0.8em}
\addtolength{\linewidth}{-0.5em}
\begin{multicols}{2}
\begin{itemize}
\item[1.] $w \Vdash p$ \textit{iff} $w \in V(p)$

\item[2.] $w \Vdash \neg p$ \textit{iff} $w \not\in V(p)$

\item[3.] $w \Vdash \phi \wedge \psi$ \textit{iff} $w \Vdash \phi$ and $w \Vdash \psi$

\item[4.] $w \Vdash \phi \vee \psi$ \textit{iff} $w \Vdash \phi$ or $w \Vdash \psi$

\item[5.] $w \Vdash \Box \phi$ \textit{iff} $\forall u \in \RB(w)$, $u \Vdash \phi$

\end{itemize}
\columnbreak
\begin{itemize}
\item[6.] $w \Vdash \Diamond \phi$ \textit{iff} $\exists u \in \RB(w)$, $u \Vdash \phi$
\item[7.] $w \Vdash [i] \phi$ \textit{iff} $\forall u \in \R_{[i]}(w)$, $u \Vdash \phi$

\item[8.] $w \Vdash \lb i \rb \phi$ \textit{iff} $\exists u \in \R_{[i]}(w)$, $u \Vdash \phi$

\item[9.] $w \Vdash \Oi \phi$ \textit{iff} $\forall u \in \opt_{\Oi}(w)$, $u \Vdash \phi$

\item[10.] $w \Vdash \ODi \phi$ \textit{iff} $\exists u \in \opt_{\Oi}(w)$, $u \Vdash \phi$

\end{itemize}
\end{multicols}
\vspace{-6pt}
\noindent 
Global truth, validity, and semantic entailment are defined as usual (see~\cite{BlaRijVen01}). We define the \textit{logic} $\ds{n}$ as the set of $\lang_{n}$ formulae valid on all $\ds{n}$-frames.
\end{definition}

\subsection{Expanded Deontic $\sti$ Frames}

In order to obtain an assortment of deontic $\sti$ characterizations accommodating the different OiC principles, we proceed in two ways: first, we define more fine-grained deontic $\sti$ operators capturing deliberative aspects of obligation, and second, we introduce a class of frame properties that change the behavior of the $\Oi$ operator when imposed on $\ds{n}$-frames.




Observe that in basic $\sti$ the choice-operator $[i]$ is a normal modal operator, which implies that $[i] \top$ is one of its validities. 
 In contrast, the more refined \textit{deliberative} $\sti$ operator---i.e. $[i]^d\phi \textit{ iff } [i]\phi \land \Diam \lnot \phi$---is non-normal and, for this reason, has been taken as defined \cite{HorBel95} (with the exception of \cite{Xu98}). (NB. For deliberative $\sti$, choices thus range over contingent state of affairs.) 
 For the same reason 
 that $\Oi\top$ is a validity of basic $\ds{n}$, we will similarly introduce two defined modalities for \textit{deliberative obligations}. Namely, we take  
$$\Oi^d \phi \textit{ iff } \Oi \phi \land \Diam \lnot \phi$$
to define a \textit{weak deliberative} obligation, expressing that an agent’s obligations can be violated (cf. \cite{Mur04,BerLyo19b}). Furthermore, we introduce 
$$\Oi^c \phi \textit{ iff } \Oi\phi \land \Diam [i]\lnot \phi$$
 as defining a \textit{strong deliberative} obligation, asserting that the obligation 
  is violable through the agent's behavior. These operators will be necessary to 
  formally capture the 
  deliberative versions of OiC in the present $\sti$ setting. 
 
Additionally, we provide
 four properties that may be imposed on $\ds{n}$-frames to change the logical behavior 
 of the $\Oi$ operator:

\vspace{7pt}
\renewcommand{\arraystretch}{1.1}
\noindent \begin{tabular}{p{1em} p{30pt} p{295pt}}
 & {\rm \textbf{D2}} & For all $w\in W$ there exists $v\in W$ s.t. $\opt_{\Oi} wv$. \\
  & {\rm \textbf{D3}} & For all $w, v\in W$, if $\RO wv$ then $\RB wv$. \\
 & {\rm \textbf{D4}} & For all $w, v, u \in W$, if $\opt_{\Oi}wv$ and $R_{\agbox}vu$, then $\opt_{\Oi}wu$.\\
 & {\rm \textbf{D5}} & For all $w \in W$, there exists a $v\in W$, such that $\opt_{\Oi} wv$ and \\
 & & for all $u \in W$, if $R_{\agbox}vu$, then $\opt_{\Oi}wu$. \\
 \end{tabular}
\vspace{3pt}

Property \textbf{D2} requires that obligations are consistent; i.e. at every moment and for every agent, there exists an ideal situation for which the agent should strive (cf. seriality in Standard Deontic Logic \cite{HilMcN13}). \textbf{D3} 
enforces that ideal worlds are confined to moments (implying 
 that 
every ideal world is realizable at its corresponding moment; cf. footnote 14). Subsequently, \textbf{D4} expresses that agent-dependent obligations are about 
 choices, thus enforcing that every ideal world coincides with an ideal choice (cf. footnote 13): i.e. 
 when `it ought to be the case for agent $i$ that' then `agent $i$ ought to see to it that' (the other direction follows from \textbf{C2} Def.~\ref{def:frames-models}). 
Lastly, \textbf{D5} states that for every agent $i$ there always exists at least one ideal 
 choice (depending on whether \textbf{D3} is adopted, this ideal choice will be guaranteed accessible by an agent or not). It must be noted that, as shown in \cite{BerLyo19b}, all four properties hold for the traditional approach to deontic $\sti$ \cite{Mur04}. We return to this in Sect.~\ref{sect:analysis}.

We define the entire class of $\sti$ logics considered in this paper as follows: 

\begin{definition}\textbf{(The logics $\dsx{n}{X}$)}\label{def:frames-models-all} Let $\mathcal{D} =\{$\textbf{D2}, \textbf{D3}, \textbf{D4}, \textbf{D5}$\}$, $n = |Ag|$ and $\mathsf{X} \subseteq \mathcal{D}$. 
 A $\ds{n}\mathsf{X}$-\textit{frame} is 
  a tuple $F = \langle W, \R_{\Box}, \{\R_{[i]} \ | \ i \in Ag\}, \{ \opt_{\Oi} \ | \ i \in Ag\} \rangle$ such that $F$ satisfies all properties of a $\ds{n}$-frame (Def.~\ref{def:frames-models}) expanded with the frame properties $\mathsf{X}$. A $\ds{n}\mathsf{X}$-\textit{model} is a tuple $(F,V)$ where $F$ is a $\ds{n}\mathsf{X}$-\textit{frame} and $V$ is a valuation function as in Def.~\ref{def:frames-models}. We define the \textit{logic} $\ds{n}\mathsf{X}$ to be the set of formulae from $\lang_{n}$ valid on all $\ds{n}\mathsf{X}$-frames. 
\end{definition}

In the following section we provide sound and complete sequent-style calculi for all logics $\ds{n}\mathsf{X}$ obtainable through Def.~\ref{def:frames-models-all}. Together with the 
defined deliberative obligation modalities 
 $\Oi^d$ and $\Oi^c$, the resulting class of calculi will suffice to capture all the deontic $\sti$ logics accommodating 
 the different OiC principles of Sect.~\ref{sect:principles}. This will be demonstrated in Sect.~\ref{sect:analysis}.



\section{Deontic STIT Calculi for Ought-implies-Can}\label{sect:calculi}


This section comprises the technical part of the paper: we introduce sound and complete sequent-style calculi $\gds{n}{X}$ for the multi-agent logics $\dsx{n}{X}$ defined in Def.~\ref{def:frames-models-all}. In what follows, we build on a simplified version of the refined labelled calculi for basic $\sti$ proposed in~\cite{LyoBer19}. 
In the present work, we modify this framework to include the deontic setting. 
 Due to space constraints, we refer to \cite{LyoBer19} for an extensive discussion on refined labelled calculi. For an introduction to sequent-style calculi in general see~\cite{Tak13}, and for labelled calculi in particular, see~\cite{Neg14,Vig00}. 
%
Labelled calculi offer a procedural, computational approach to making explicit semantic arguments. This approach not only allows for a precise understanding of the logical relationships between the different OiC readings and corresponding logics, but can additionally be harnessed to construct counter-models 
  confirming the independence of certain OiC principles. We will demonstrate this in Sect.~\ref{sect:analysis}. 
 
\begin{definition}Let $Lab := \{ x, y, z, ... \}$ be a denumerable set of \textit{labels}. The language of our calculi consists of \emph{sequents} $\Lambda$, which are syntactic objects of the form $\rel \vdash \Gamma$. 
 $\rel$ and $\Gamma$ are defined via the following BNF grammars:
$$\rel ::= \varepsilon \ | \ \R_{\Box}xy \ | \ \R_{[i]}xy \ | \ \opt_{\Oi}xy \ | \ \rel, \rel \qquad\quad \Gamma ::= \varepsilon \ | \ x:\phi \ | \ \Gamma, \Gamma$$
with $i \in Ag$, $\phi \in \lang_{n}$, and $x,y \in Lab$. 
\end{definition}

We refer to $\rel$ as the 
  \emph{antecedent} of $\Lambda$ 
 and to $\Gamma$ as 
 the \emph{consequent} of $\Lambda$. We use $\mathcal{R}$, 
 $\mathcal{R}'$, 
  $\ldots$ to denote strings generated by the top left grammar and refer to formulae (e.g. $\R_{[i]}xy$ and $\opt_{\Oi}xy$) occurring in such strings as \emph{relational atoms}. 
We use $\Gamma$, 
 $\Gamma'$, 
  $\ldots$ to denote strings generated by the top right grammar and refer to formulae (e.g. $x : \phi$) occurring in such strings as \emph{labelled formulae}. We take the comma operator to commute and associate in $\rel$ and $\Gamma$ (i.e. $\rel$ and $\Gamma$ are multisets) and read its presence in $\rel$ and $\Gamma$, respectively, as a conjunction and a disjunction (cf.~Def.~\ref{def:sequent-semantics}). We let $\varepsilon$ represent the \emph{empty string}.
  \footnote{The empty string $\varepsilon$ serves as an identity element for comma (e.g. $R_{\Box}xy, \varepsilon \vdash x : p, \varepsilon, y : q$ identifies with $R_{\Box}xy \vdash x : p, y : q$). If $\varepsilon$ is the entire antecedent or consequent, it is left empty by convention (e.g. $\varepsilon \vdash \Gamma$ identifies with $\vdash \Gamma$). In what follows, it suffices to leave $\varepsilon$ implicit.} 
Last, we use $Lab(\rel \vdash \Gamma)$ 
to represent the set of labels contained in $\rel \vdash \Gamma$. 
 
\begin{figure}[t]
\noindent
\begin{small}
\begin{center}
\begin{tabular}{c c} 
\AxiomC{ }
\RightLabel{$\id$}
\UnaryInfC{$\rel \vdash x:p, x:\neg p, \Gamma$}
\DisplayProof

&

\AxiomC{$\rel \vdash x: \phi, w : \psi, \Gamma$}
\RightLabel{$\disr$}
\UnaryInfC{$\rel \vdash x: \phi \vee \psi, \Gamma$}
\DisplayProof
\end{tabular}
\end{center}

\begin{center}
\begin{tabular}{c @{\hskip 2em} c}
\AxiomC{$\rel \vdash x: \phi, \Gamma$}
\AxiomC{$\rel \vdash x: \psi, \Gamma$}
\RightLabel{$\conr$}
\BinaryInfC{$\rel \vdash x: \phi \wedge \psi, \Gamma$}
\DisplayProof

&

\AxiomC{$\rel,\R_{[1]}x_{1}y, ..., \R_{[n]}x_{n}y \vdash \Gamma$}
\RightLabel{$\ioa^{\dag_{2}}$}
\UnaryInfC{$\rel \vdash \Gamma$}
\DisplayProof
\end{tabular}
\end{center}

\begin{center}
\begin{tabular}{c c}
\AxiomC{$\rel, R_{\Box} xy \vdash y: \phi, \Gamma$}
\RightLabel{$\settr^{\dag_{1}}$}
\UnaryInfC{$\rel \vdash x: \Box \phi, \Gamma$}
\DisplayProof

&

\AxiomC{$\rel \vdash x: \Diamond \phi, y: \phi, \Gamma$}
\RightLabel{$\settdiar^{\dag_{3}}$}
\UnaryInfC{$\rel \vdash x: \Diamond \phi, \Gamma$}
\DisplayProof
\end{tabular}
\end{center}

%
%

\begin{center}
\begin{tabular}{c @{\hskip 2em} c}

\AxiomC{$\rel, \R_{[i]}xy \vdash y: \phi, \Gamma$}
\RightLabel{$\stitr^{\dag_{1}}$}
\UnaryInfC{$\rel \vdash x: [i] \phi, \Gamma$}
\DisplayProof

&

\AxiomC{$\rel \vdash x: \agdia \phi, y:\phi, \Gamma$}
\RightLabel{$\stitdiar^{\dag_{4}}$}
\UnaryInfC{$\rel \vdash x: \agdia \phi, \Gamma$}
\DisplayProof

\end{tabular}
\end{center}

\begin{center}
\begin{tabular}{c @{\hskip 2em} c}
\AxiomC{$\rel, \opt_{\Oi}xy \vdash y : \phi, \Gamma$}
\RightLabel{$\Oir^{\dag_{1}}$}
\UnaryInfC{$\rel \vdash x : \Oi \phi, \Gamma$}
\DisplayProof

&

\AxiomC{$\rel, \opt_{\Oi}xy \vdash x : \ominus_{i} \phi, y : \phi, \Gamma$}
\RightLabel{$\Odiamrule$}
\UnaryInfC{$\rel, \opt_{\Oi} xy \vdash x : \ominus_{i} \phi, \Gamma$}
\DisplayProof
\end{tabular}
\end{center}
\begin{center}
\begin{tabular}{c}
\AxiomC{$\rel, \opt_{\Oi} xz, \opt_{\Oi} yz \vdash \Gamma$}
\RL{$(\mathsf{D1}_{i})^{\dag_{3}}$}
\UIC{$\rel, \opt_{\Oi} xz \vdash \Gamma$}
\DisplayProof
\end{tabular}
\end{center}
\end{small}

\hrule
\caption{The calculi $\gds{n}{}$ (with $n = |Ag|$). 
 $\dag_{1}$ on $\settr$, $\stitr$, and $\Oir$ indicates that $y$ is an eigenvariable, i.e. $y$ does not occur in the rule's conclusion. $\dag_2$ on $\ioa$ states that $y$ is an eigenvariable and for all $i \in \{1, \ldots, n\}$, $x_{i} \sim^{\rel}_{\Diamond} x_{i+1}$ (see Def.~\protect\ref{def:sett-path}). $\dag_3$ on $\settdiar$ and $\done$ and $\dag_4$ on $\stitdiar$ state, respectively, that $x \sim^{\rel}_{\Diamond} y$ and $x \sim^{\rel}_{i} y$ (see Def.~\protect\ref{def:sett-path} and Def.~\protect\ref{def:stit-path}). We have $\stitr$, $\stitdiar$, $\Oir$, $\Odiamrule$, and $(\mathsf{D1}_{i})$ rules for each $i \in Ag$.}
\label{fig:base-calculus}
\end{figure}
  
 The calculus $\gds{n}{}$ for the minimal deontic $\sti$ logic $\ds{n}$ (with $n \in \mathbb{N}$) 
 is shown in Fig.~\ref{fig:base-calculus}. Intuitively, 
$\gds{n}{}$ can be seen as a 
transformation 
of the semantic clauses of Def.~\ref{def:semantics} and $\ds{n}$-frame properties of Def.~\ref{def:frames-models} into inference rules. For example, the $\id$ rule encodes the fact that either a propositional atom $p$ holds at a world in a $\ds{n}$-model, or it does not (recall that a comma in the consequent reads disjunctively). The rules $\ioa$ and $(\mathsf{D1}_{i})$ encode, respectively, condition \textbf{C3} 
(i.e. independence of agents) and 
condition \textbf{D1} of Def.~\ref{def:frames-models}. A particular feature of refinement, is that we can incorporate the semantic behavior of modalities into their corresponding rules. For instance, the side condition $\dag_4$ of the $(\langle i\rangle)$ rule integrates the fact that $\langle i\rangle$ is semantically characterized as an equivalence relation.  These side conditions---including those for the rules $\settdiar$, $\stitdiar$ and $\done$---rely on the notion of a \emph{$\Diamond$-} and \emph{$\langle i \rangle$-path}. 

%
 
\begin{definition}\textbf{($\langle i\rangle$-path)}\label{def:stit-path} Let $x \sim_i y \in \{R_{[i]}xy, R_{[i]}yx\}$ and $\Lambda = \rel \vdash \Gamma$. An \emph{$\langle i\rangle$-path} of relational atoms from a label $x$ to $y$ occurs in $\Lambda$ (written as $x \sim^{\rel}_{i} y$) iff $x = y$, $x \sim_i y$, or there exist labels $z_{j}$ ($j \in \{1,\ldots,k\}$) such that $x \sim_i z_{1}, \ldots, z_{k} \sim_i y$ occurs in $\rel$.
\end{definition}

\begin{definition}\textbf{($\Diam$-path)}\label{def:sett-path} Let $x \sim_{\Diam} y \in \{R_{\Box}xy, R_{\Box}yx\} \cup \{R_{[i]}xy, R_{[i]}yx \ | \ i\in Ag \}$, and $\Lambda = \rel \vdash \Gamma$. An \emph{$\Diam$-path} of relational atoms from a label $x$ to $y$ occurs in $\Lambda$ (written as $x \sim^{\rel}_{\Diam} y$) iff $x = y$, $x \sim_{\Diam} y$, 
 or there exist labels $z_{j}$ ($j \in \{1,\ldots,k\}$) such that $x \sim_{\Diam} z_{1}, \ldots, z_{k} \sim_{\Diam} y$ occurs in $\rel$.

%
\end{definition}

The definition of an $\langle i \rangle$- and $\Diam$-path captures a notion of reachability that simulates the fact that $\RI$ and $\RB$ are equivalence relations. Moreover, $\Diam$-paths also incorporate the fact that choices are subsumed under moments (cf.~\textbf{C2} of Def.~\ref{def:frames-models}). Observe that the $\Diam$-path condition on $\ioa$ indicates that `independence of agents' can only be applied to choices that occur at the same moment. One of the advantages of using such paths as side conditions is that it allows us to reduce the number of rules in our calculi~\cite{LyoBer19}.


Fig.~\ref{fig:structural-rules} contains four additional structural rules with which the base calculi $\gds{n}{}$ can be extended. As their names suggest, these rules simulate their respective frame properties (cf. Def.~\ref{def:frames-models-all}). In doing so, we obtain calculi for the logics $\dsx{n}{X}$. As an example,  the logic $\dsx{n}{\{\textbf{D2,D4}\}}$ corresponds to the calculus
  $\gds{n}{\{\dtwo,\dfour \ | \ \textit{i} \in \textit{Ag} \}}$ (henceforth, we write $\gds{n}{\{D2_i,D3_i\}}$). 


\begin{definition}\textbf{(The calculi $\gds{n}{X}$)}\label{def:gdsnx-all} Let $\dsx{n}{X}$ be a logic from Def.~\ref{def:frames-models-all}. Let $n = |Ag| \in \mathbb{N}$ and $\mathsf{X} \subseteq \{\textbf{D2}, \textbf{D3}, \textbf{D4}, \textbf{D5}\}$. We define $\gds{n}{X}$ to consist of $\gds{n}{}$ extended with $(\mathsf{DK}_{i})$, if $\textbf{DK} \in \mathsf{X}$ (with $K \in \{2,3,4,5\}$) for all $i \in Ag$.
\end{definition}

 We point out that the first-order condition \textbf{D5} (Def.~\ref{def:frames-models}) is a \textit{generalized geometric axiom}. In~\cite{Neg14}, it was shown that properties of this form require \textit{system of rules} in their corresponding calculi. We adopt this approach in our calculi as well and use $\dfive$ to denote the system of rules $\langle \dfivei, \dfiveii\rangle$ (see Fig.~\ref{fig:structural-rules}). The global restriction $\ddag$ imposed on applying $\dfive$ is that, although we may use $\dfivei$ wherever, if we use $\dfiveii$ we must also use $\dfivei$ further down in the derivation. 
 In Sect.~\ref{sect:analysis}, Ex.~\ref{ex:transforming-d5-tod2d3-derivation} demonstrates an application of $\dfive$. 

To confirm soundness and completeness for our calculi---thus 
demonstrating an equivalence between the semantics ($\dsx{n}{X}$) 
 and proof-theory ($\gds{n}{X}$) 
  of our logics---we  need to provide a semantic interpretations of sequents: 

\begin{definition}\textbf{(Sequent Semantics)}\label{def:sequent-semantics} Let $M$ be a $\dsx{n}{X}$-model with domain $W$ and $I$ an \emph{interpretation function} mapping labels to worlds; i.e. $I{:} \ Lab \mapsto W$. A sequent $\Lambda = \rel \vdash \Gamma$ is \emph{satisfied} in $M$ with $I$ (written, $M,I \models \Lambda$) iff for all relational atoms $R_{\alpha} xy \in \rel$ (where $\alpha \in \{\Box\} \cup \{ [i], \Oi \ | \ i\in Ag\}$), if $R_{\alpha}I(x)I(y)$ holds in $M$, then there exists a $z : \phi \in \Gamma$ such that $M, I(z) \Vdash \phi$. 
$\Lambda$ is \emph{valid relative to $\mathsf{DS}_{n}\mathsf{X}$} iff it is satisfiable in any 
 $\dsx{n}{X}$-model $M$ with any  %
 $I$.
\end{definition}

%
%
%
%
%
%
%
%
%
%

\begin{figure}[t]

\noindent
\begin{center}
\resizebox{\columnwidth}{!}{
\begin{tabular}{c c}
\AxiomC{$\rel, \opt_{\Oi}xy \vdash \Gamma$}
\RightLabel{$(\mathsf{D2}_i)^{\dag_{1}}$}
\UnaryInfC{$\rel \vdash \Gamma$}
\noLine
\UIC{}
\noLine
\UIC{}
\noLine
\UIC{}
\noLine
\UIC{}
\AxiomC{$\rel, \opt_{\Oi}xy, \RB xy \vdash \Gamma$}
\RightLabel{$(\mathsf{D3}_{i})$}
\UnaryInfC{$\rel, \RO xy \vdash \Gamma$}
\noLine
\UIC{}
\noLine
\UIC{}
\noLine
\UIC{}
\noLine
\UIC{}
\noLine
\BIC{$\rel, \opt_{\Oi}xy, \R_{\Oi}xz \vdash \Gamma$} 
\RightLabel{$(\mathsf{D4}_{i})^{\dag_{2}}$}
\UnaryInfC{$\rel, \opt_{\Oi}xy \vdash \Gamma$} 
\DisplayProof

&

\AxiomC{$\rel',\opt_{\Oi}xz \vdash \Gamma'$}
\RightLabel{$\dfiveii^{\dag_{2}}$}
\UIC{$\rel' \vdash \Gamma'$} 
\noLine
\UIC{$\vdots$}
\noLine
\UIC{$\rel, \opt_{\Oi}xy \vdash \Gamma$}
\RightLabel{$\dfivei^{\dag_1}$}
\UnaryInfC{$\rel \vdash \Gamma$}
\DisplayProof
\end{tabular}
}
\end{center}

\hrule
\caption{The Deontic Structural Rules. Condition $\dag_{1}$ on $\dtwo$ and $\dfivei$ states that $y$ is a eigenvariable. Condition $\dag_{2}$ on $\dfour$ and $\dfiveii$ indicates that $y \sim^{\rel}_{i} z$ (Def.~\ref{def:stit-path}). Last, we let $\dfive^{\ddag}$ be $\langle \dfivei, \dfiveii\rangle$ with $\ddag$ the global restriction (mentioned below), and have $\dtwo, \dthree, \dfour, \dfive$ rules for each $i \in Ag$.}\label{fig:structural-rules}
\end{figure}

\begin{theorem}[Soundness and Completeness of $\gds{n}{X}$]\label{thm:adequacy-dstitn} A sequent $\Lambda$ is derivable in $\gds{n}{X}$ iff it is valid relative to $\mathsf{DS}_{n}\mathsf{X}$. 
\end{theorem}

\begin{proof} Follows from Thm.~\ref{thm:sound-dstitn} and~\ref{thm:compl-dstitn}. See the Appendix~\ref{appendix} for details.
\end{proof}


%

\section{A formal analysis of Deontic STIT and OiC}\label{sect:analysis}

In this section, we put our $\gds{n}{X}$ calculi to work. First, we make use of our calculi to organize our logics in terms of their strength---observing which are equivalent, distinct, or subsumed by another. Second, we discuss the logical (in)dependencies between our various OiC principles by confirming the minimal logic in which each principle is validated.



\subsection{A Taxonomy of Deontic $\sti$ Logics}

In Fig.~\ref{fig:oic-relations}, a lattice is provided ordering the sixteen deontic $\sti$ calculi of Def.~\ref{def:gdsnx-all} on the basis of their respective strength (reflexive and transitive edges are left implicit). 
We consider a calculus $\gds{n}{X}$ stronger than another calculus $\gds{n}{Y}$ whenever the former generates at least the same set of theorems as the latter. Consequently, 
the lattice simultaneously orders the deontic $\sti$ logics of Def.~\ref{def:frames-models-all}, generated by these calculi, on the basis of their expressivity. In Fig.~\ref{fig:oic-relations}, the calculi are ordered bottom-up: 
 $\gds{n}{}$ is the weakest system, generating the smallest logic subsumed by all others, whereas $\gds{n}{\{D2_i,D3_i,D4_i\}}$ is the strongest calculus with its logic subsuming all others. Notice that the latter calculus generates the
  traditional deontic $\sti$ logic of 
 \cite{Hor01,Mur04}. To determine the existence of a directed edge from one calculus $\gds{n}{X}$ to another $\gds{n}{Y}$ in the lattice, we need to show that  every derivation in the former can be transformed into a derivation in the latter. As an example of this procedure, we consider the edge from $\gds{n}{\{D3_i,D5_i\}}$ to $\gds{n}{\{D2_i,D3_i,D4_i\}}$. 

\begin{example}\label{ex:transforming-d5-tod2d3-derivation} To transform a $\gds{n}{\{D3_i,D5_i\}}$-derivation into a derivation of $\gds{n}{\{D2_i,D3_i,D4_i\}}$, it suffices to show that each instance of $\dfivei$ and $\dfiveii$ can be replaced, respectively, by instances of $\dtwo$ and $\dfour$. For example: 
\vspace{0.2cm}
\begin{flushleft}
\begin{minipage}{0.45\linewidth}
\scalebox{0.655}{
\AXC{$\RB xy, \RO xy, \RI yz, \RO xz \vdash z: \neg \phi,..., z: \phi  $}\RL{$(\ODi)$}\LL{}
\UIC{$\RB xy, \RO xy, \RI yz, \RO xz \vdash x: \ODi \neg \phi,..., z: \phi    $}\RL{$\dfiveii$}
\UIC{$\RB xy, \RO xy, \RI yz \vdash x: \ODi \neg \phi,..., z: \phi   $}\RL{$([i])$}
\UIC{$\RB xy, \RO xy \vdash x: \ODi \neg \phi,..., y: [i] \phi  $}\RL{$(\Diamond)$}
\UIC{$\RB xy, \RO xy \vdash x: \ODi \neg \phi, x : \Diam [i] \phi $}\RL{$\dthree\quad\quad\quad$ {\huge $\leadsto$}}
\UIC{$\RO xy \vdash x: \ODi \neg \phi, x : \Diam [i] \phi $}\RL{$\dfivei$}
\UIC{$\vdash x: \ODi \neg \phi, x : \Diam [i] \phi $}\RL{$(\lor) $}\LL{$\quad\quad\quad$}
\UIC{$\vdash x: \ODi \neg \phi \lor \Diam [i] \phi$}
\RightLabel{=}
\dottedLine
\UnaryInfC{$\vdash x : \Oi \phi \rightarrow \Diam[i]\phi$}
\DisplayProof
%
%
%
}
\end{minipage}
\hspace{0.6cm}
\begin{minipage}{0.32\linewidth}
\scalebox{0.655}{
\AXC{$\RB xy, \RO xy, \RI yz, \RO xz \vdash z: \neg \phi,..., z: \phi  $}\RL{$(\ODi)$}\LL{}
\UIC{$\RB xy, \RO xy, \RI yz, \RO xz \vdash x: \ODi \neg \phi,..., z: \phi    $}\RL{$\dfour$}
\UIC{$\RB xy, \RO xy, \RI yz \vdash x: \ODi \neg \phi,..., z: \phi   $}\RL{$([i])$}
\UIC{$\RB xy, \RO xy \vdash x: \ODi \neg \phi,..., y: [i] \phi  $}\RL{$(\Diamond)$}
\UIC{$\RB xy, \RO xy \vdash x: \ODi \neg \phi, x : \Diam [i] \phi $}\RL{$\dthree$}
\UIC{$\RO xy \vdash x: \ODi \neg \phi, x : \Diam [i] \phi $}\RL{$\dtwo$}
\UIC{$\vdash x: \ODi \neg \phi, x : \Diam [i] \phi $}\RL{$(\lor) $}\LL{$\quad\quad\quad$}
\UIC{$\vdash x: \ODi \neg \phi \lor \Diam [i] \phi$}
\RightLabel{=}
\dottedLine
\UnaryInfC{$\vdash x : \Oi \phi \rightarrow \Diam[i]\phi$}
\DisplayProof
}
\end{minipage}
\end{flushleft}
\vspace{0.2cm}
The non-existence of a directed edge in the opposite direction is implied by the fact that $\gds{n}{\{D2_i,D3_i,D4_i\}} \vdash\Oi\phi\rightarrow\Oi[i]\phi$ and $\gds{n}{\{D3_i,D5_i\}} \not\vdash\Oi\phi\rightarrow\Oi[i]\phi$. The latter is shown 
 through failed proof search (See Ex.~\ref{ex:underivable-oina} for an illustration of how failed proof-search can be used to determine underivability.)
\end{example}
To determine that two calculi $\gds{n}{X}$ and $\gds{n}{Y}$ are equivalent (i.e. $\gds{n}{X} \equiv \gds{n}{Y}$), thus implying that the associated logics are identical, one shows that every derivation in the former can be transformed into a derivation in the latter, and vice-versa. Last, to prove that two calculi $\gds{n}{X}$ and $\gds{n}{Y}$ 
 are independent---yielding incomparable logics---it is sufficient to show that there exist formulae $\phi$ and $\psi$ such that $\gds{n}{X} \vdash \phi$, $\gds{n}{Y} \not\vdash \phi$, $\gds{n}{Y} \vdash \psi$, and $\gds{n}{X} \not\vdash \psi$. We come back to this in the following subsection when we consider an example of an underivable OiC formula.

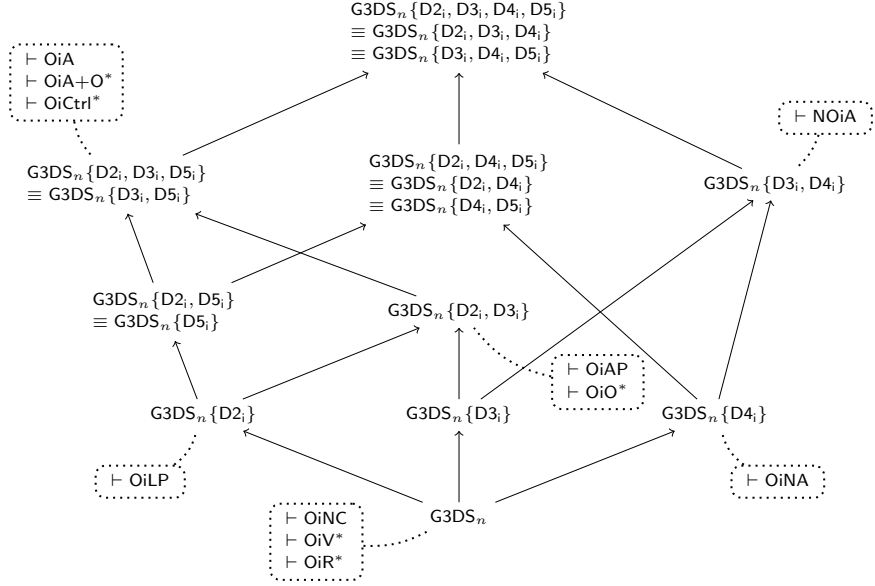
\begin{figure}

\begin{center}
\scalebox{0.945}{
\begin{scriptsize}

\begin{tikzpicture}

\node[] (234) {\begin{tabular}{l}
$\gds{n}{\{D2_i,D3_i,D4_i,D5_i\}}$\\
$\equiv\gds{n}{\{D2_i,D3_i,D4_i\}}$\\
$\equiv\gds{n}{\{D3_i,D4_i,D5_i\}}$
\end{tabular}};


\node[] (24) [below=of 234] {\begin{tabular}{l}
$\gds{n}{\{D2_i,D4_i,D5_i\}}$\\
$\equiv\gds{n}{\{D2_i,D4_i\}}$\\
$\equiv\gds{n}{\{D4_i,D5_i\}}$
\end{tabular}};

\node[] (35) [left=of 24, xshift=-18pt] 
{\begin{tabular}{l}
$\gds{n}{\{D2_i,D3_i,D5_i\}}$\\
$\equiv\gds{n}{\{D3_i,D5_i\}}$\\
\end{tabular}};

\node[draw,thick,dotted,rectangle, color=black,rounded corners, inner xsep=0pt, inner ysep=2pt,] (oic35) [above=of 35, xshift=-20pt, yshift=-15pt] {\begin{tabular}{l}
$\vdash\oia$ \\
$\vdash\oiao^{\ast}\!\!$ \\
$\vdash\oictrl^{\ast}\!\!$
\end{tabular}};

\node[] (34) [right=of 24,xshift=21pt] {$\gds{n}{\{D3_i,D4_i\}}$};

\node[] (23) [below=of 24] {$\gds{n}{\{D2_i,D3_i\}}$};

\node[draw, thick, dotted,rectangle, color=black,rounded corners, inner xsep=0pt, inner ysep=2pt,] (oic23) [below=of 23, xshift=55pt, yshift=18pt] {\begin{tabular}{l}
$\vdash \oipp$\\
$\vdash \oio^{\ast}$
\end{tabular}};

\node[] (5) [left=of 23, xshift=-20pt] 
{\begin{tabular}{l}
$\gds{n}{\{D2_i,D5_i\}}$\\
$\equiv\gds{n}{\{D5_i\}}$\\
\end{tabular}};

\node[draw, thick, dotted,rectangle, color=black,rounded corners, inner xsep=0pt, inner ysep=2pt,] (oic5) [above=of 34, xshift=20pt, yshift=-15pt] {\begin{tabular}{l}
$\vdash\mathsf{NOiA}$
\end{tabular}};

\node[] (3) [below=of 23] {$\gds{n}{\{D3_i\}}$};

\node[] (2) [left=of 3, xshift=-25pt] {$\gds{n}{\{D2_i\}}$};

\node[draw,thick,dotted,rectangle, color=black,rounded corners, inner xsep=0pt, inner ysep=2pt,] (oic2) [below=of 2, xshift=-25pt, yshift=15pt] {\begin{tabular}{l}
$\vdash\oilp$
\end{tabular}};

\node[] (4) [right=of 3, xshift=25pt] {$\gds{n}{\{D4_i\}}$};

\node[draw,thick,dotted,rectangle, color=black,rounded corners, inner xsep=0pt, inner ysep=2pt,] (oic4) [below=of 4, xshift=25pt, yshift=15pt] {\begin{tabular}{l}
$\vdash\oina$
\end{tabular}};

\node[] (0) [below=of 3
] {$\gds{n}{}$};

\node[draw,thick,dotted,rectangle, color=black,rounded corners, inner xsep=0pt, inner ysep=2pt,] (oic0) [left=of 0, xshift=5pt, yshift=-10pt] {\begin{tabular}{l}
$\vdash\oinc$\\
$\vdash\oiv^{\ast}$ \\
$\vdash\oir^{\ast}$ \\
\end{tabular}};

\path[<-,draw] (234) edge (35);

\path[<-,draw] (234) edge (24);

\path[<-,draw] (234) edge (34);

\path[<-,draw] (35) edge (23);

\path[<-,draw] (35) edge (5);

\path[<-,draw] (24) edge (5);

\path[<-,draw] (24) edge (4);

\path[<-,draw] (34) edge (3);

\path[<-,draw] (34) edge (4);

\path[<-,draw] (23) edge (2);

\path[<-,draw] (23) edge (3);

\path[<-,draw] (5) edge (2);

\path[<-,draw] (2) edge (0);

\path[<-,draw] (3) edge (0);

\path[<-,draw] (4) edge (0);

\path[-
,dotted,draw,thick, bend right=15
] (oic35) edge (35);

\path[-
,
dotted,draw, thick, bend left=15
] (oic5) edge (34);

\path[-
,dotted,draw, thick, bend right=20
] (oic2) edge (2);

\draw[-
,dotted,draw, thick, bend right=15
] (oic0) edge (0);

\path[-
,dotted,draw, thick,bend left=20
] (oic4) edge (4);

\path[-
,dotted,draw, thick, bend left=20
 ] (oic23) edge (23);
\end{tikzpicture}

\end{scriptsize}
}
\end{center}
\vspace{-0.3cm}
\caption{The lattice of deontic $\sti$ calculi. 
 Directed edges point from weaker calculi to stronger calculi, 
 consequently
 ordering the corresponding logics w.r.t. their expressivity (reflexive and transitive edges are left implicit). 
 We use $\equiv$ to denote equivalent calculi. 
 Dotted nodes show which calculi should at least be adopted to make the indicated OiC principles theorems (for the final OiC formalizations see Fig.~\ref{fig:oic-calculi}). 
 }\label{fig:oic-relations}
\end{figure}

\subsection{Logical (In)Dependencies of OiC Principles}

Fig.~\ref{fig:oic-relations} also represents which deontic $\sti$ calculi should at least be adopted to make certain OiC principles theorems of the corresponding logics. These principles were initially formalized in Sect.~\ref{sect:principles}. However, as discussed in Sect.~\ref{sect:logics}, in order to formally represent \textit{deliberative} readings of OiC in a normal modal setting, we must replace the initial antecedent $\Oi\phi$ with its deliberative correspondent $\Oi^d\phi$ in $\oiv,\oir,\oio,\oiao$ and with $\Oi^c\phi$ in $\oictrl$. The final list of OiC formalizations is presented in Fig.~\ref{fig:oic-calculi}. Although for now the above suffices---i.e. the approach being in line with the traditional treatment of deliberative agency 
\cite{BelPerXu01,Hor01,HorBel95}---the solution may be considered \textit{ad hoc}. 
 We note that these deliberative canons may alternatively be captured as follows: 
 (i) through characterizing deliberation directly in the logic, taking $\Oi^d$ and $\Oi^c$ as primitive operators (cf.~\cite{Xu98}), or (ii) through 
  characterizing contingency via the use of sanction constants 
  (cf.~\cite{AndMoo57}). We leave this to future work.
 


In Ex.~\ref{ex:transforming-d5-tod2d3-derivation}, we saw that $\oia$ is derivable in both $\gds{n}{\{D2_i,D3_i,D4_i\}}$ and $\gds{n}{\{D3_i,D5_i\}}$. What is more, since $\Oi[i]\phi\rightarrow \Oi\phi$ is already a theorem of $\gds{n}{}$, we find that the weaker logic generated by $\gds{n}{\{D3_i,D5_i\}}$ already suffices to accommodate 
OiC 
 of the traditional deontic $\sti$ setting~\cite{Hor01}, that is, $\gds{n}{\{D3_i,D5_i\}} \vdash \Oi[i]\phi\rightarrow \Diam [i]\phi$. 
We emphasize 
 that only through the addition of $\mathsf{D4_i}$ do we restore the position advocated by Horty in \cite{Hor01} (cf. footnote 2). 
 Namely, by adding $\mathsf{D4_i}$ to a calculus, the distinction between $\Oi$ and $\Oi[i]$ collapses---i.e. $\gds{n}{\{D4_i\}}\vdash \Oi\phi \equiv \Oi[i]\phi$---and the agent-dependent obligation operator will demonstrate the same logical behavior as the interpretation of obligation restricted to complete choices; i.e. the `dominance ought'. (NB. In 
  \cite{BerLyo19b} it was 
   shown that the relational characterization of $\Oi$ in $\mathsf{DS}_n \{\textbf{D2},\textbf{D3},\textbf{D4}\}$ is equivalent to the logic of `dominance ought' \cite{Hor01,Mur04}.) 

\begin{figure}
\begin{center}
{\small
\scalebox{0.87}{
\def\arraystretch{1.4}{
\begin{tabular}{l l  | l l  } 
 $\!\!\!\!\!\!\!\!\gds{n}{\{D2_i\}} \vdash \Oi \phi \rightarrow \lnot\! \Oi\!\lnot \phi$   & $\mathsf{OiLP}\!\! $ &
 $\gds{n}{\{D2_i,D3_i\}}\vdash\Oi^d \phi \rightarrow (\Diam \phi \land \Diam \lnot \phi)$ & $\mathsf{OiO}^{\ast} $ 
\\
 $\!\!\!\!\!\!\!\!\gds{n}{\{D2_i,D3_i\}}\vdash\Oi \phi \rightarrow \Diam \phi$  & 
$\oipp\!\!\! $
& $\gds{n}{\{D3_i,D5_i\}}\vdash\Oi^d \phi \rightarrow (\Diam [i] \phi \land \Diam \phi \land \Diam \lnot \phi)\!\!\!\!\!\!$ & $\mathsf{OiA{+}O}^{\ast}\!\!\! $
\\
 $\!\!\!\!\!\!\!\!\gds{n}{\{D3_i,D5_i\}}\vdash\Oi \phi \rightarrow \Diam [i] \phi\!\!\!\!$ 
 & $\mathsf{OiA} $  
 & $\gds{n}{\{D3_i,D5_i\}}\vdash\Oi^c \phi \rightarrow (\Diam [i] \phi \land \Diam [i] \lnot \phi) \!\!\!\!\!$ 
&  $\mathsf{OiCtrl}^{\ast} $ \\
$\!\!\!\!\!\!\!\!\gds{n}{}\vdash\Oi^d \phi \rightarrow \Diam \lnot \phi$ 
& $\mathsf{OiV}^{\ast} $
 & $\gds{n}{}\vdash\Oi\phi \rightarrow \Oi\Diam\phi$  &  $\oinc$\\
 $\!\!\!\!\!\!\!\!\gds{n}{}\vdash\Oi^d \phi \rightarrow \Diam [i]\lnot [i] \phi \!\!\!$ & $\mathsf{OiR}^{\ast} $  
 & $\gds{n}{\{D4_i\}}\vdash\Oi\phi \rightarrow \Oi \Diam [i] \phi$ &  $\oina$\\
\end{tabular}
}
}
}
\end{center}

\hrule
\caption{$\sti$ formalizations of OiC, with the 
 minimal $\gds{n}{X}$
  calculi entailing them. 
}
\label{fig:oic-calculi}
\end{figure}

From a philosophical perspective, 
 Fig.~\ref{fig:oic-relations} gives rise to what we will call the \emph{endorsement principle} of the philosophy of OiC. Namely, the ordering  of calculi 
 tells us which endorsements of which OiC readings will logically commit us to endorsing other OiC readings (within the logical realm of agential choice). For instance, endorsing $\oia$ tells us that we must also endorse the weaker 
  $\oilp$ and $\oipp$ since they are logically
   entailed in the minimal calculus 
  for $\oia$.

Furthermore, the 
 taxonomy of deontic $\sti$ logics shows 
  which readings of OiC are independent from one another. In particular, we note that the normative principle $\oina$ is strictly 
   independent of $\oia, \oilp, \oipp$. An advantage of 
   the present proof theoretic approach is that we can constructively prove why certain readings of OiC fail to entail one another (relative to their 
    calculi):

\begin{example}\label{ex:underivable-oina} 
To show that $\oina$ is not entailed by $\oilp$
 in $\gds{1}{\{D2_1\}}$ one attempts to prove an 
  instance of $\oina$ via bottom-up proof-search (left): 
\vspace{-0.2cm}
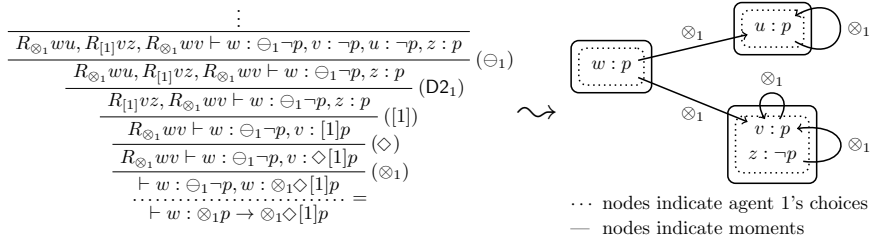
\begin{flushleft}
\begin{minipage}{0.45\linewidth}
\scalebox{0.72}{
\AXC{$\vdots$}
\UIC{$R_{\otimes_1} wu, R_{[1]} vz, R_{\otimes_1} wv \vdash w: \ominus_1 \lnot p, v: \lnot p, u: \lnot p, z: p$}\RL{$(\ominus_1)$}
\UIC{$R_{\otimes_1} wu, R_{[1]} vz, R_{\otimes_1} wv \vdash w: \ominus_1 \lnot p, z: p$}\RL{$(\mathsf{D2}_1)$}
\UIC{$R_{[1]} vz, R_{\otimes_1} wv \vdash w: \ominus_1 \lnot p, z: p$}\RL{$([1])\quad\quad\quad\quad\quad\!$ {\huge $\leadsto$}}
\UIC{$R_{\otimes_1} wv \vdash w:  \ominus_1 \lnot p, v:[1] p$}\RL{$(\Diam)$}
\UIC{$R_{\otimes_1} wv \vdash w: \ominus_1 \lnot p, v:\Diam [1] p$}\RL{$(\otimes_1)$}
\UIC{$\vdash w: \ominus_1 \lnot p, w:\otimes_1\Diam [1] p$}\RightLabel{=}
\dottedLine
\UIC{$\vdash w: \otimes_1 p\rightarrow \otimes_1\Diam [1] p$}
\DisplayProof
}
\end{minipage}
\hspace{1.7cm}
\begin{minipage}{0.32\linewidth}
\scalebox{0.72}{
\begin{tabular}{p{5.5cm}}
\begin{tikzpicture}
\node[] (w) [] {$w: p$};
\node[] (u) [right=of w,yshift=0.7cm,xshift=1cm] {$u:p$};
\node[] (v) [right=of w,yshift=-1.2cm,xshift=1cm] {$v:p$};
\node[] (z) [below=of v,yshift=1cm] {$z:\lnot p$};
\node[draw,dotted, thick, rounded corners, inner xsep=3pt, inner ysep=3pt, fit=(w)] (c1) {};
\node[draw, rounded corners, thick,  inner xsep=7pt, inner ysep=7pt, fit=(w)] (m1) {};
\node[draw,dotted, thick, rounded corners, inner xsep=3pt, inner ysep=3pt, fit=(u)] (c2) {};
\node[draw, rounded corners, thick,  inner xsep=7pt, inner ysep=7pt, fit=(u)] (m2) {};
\node[draw,dotted, thick, rounded corners, inner xsep=3pt, inner ysep=1.5pt, fit=(v) (z)] (c3) {};
\node[draw, rounded corners, thick,  inner xsep=7pt, inner ysep=7pt, fit=(v) (z)] (m3) {};
\draw[->,thick] (w) edge node[above] {$\otimes_1$} (u);
\draw[->,thick] (w) edge node[below] {$\otimes_1$} (v);
\draw[reflexive right,thick] (u) edge node[right] {$\otimes_1$} (u);
\draw[reflexive above,thick] (v) edge node[above] {$\otimes_1$} (v);
\draw[->,thick,label=right:$\otimes_1$] (z)  .. controls +(1.5,-0.2) and +(1.5,0.2) ..
node[right] {$\otimes_1$}
(v);
\end{tikzpicture}
\\
$\cdots$ nodes indicate agent $1$'s choices\\
\textemdash~ nodes indicate moments\\
\end{tabular}
}
\end{minipage}
\end{flushleft}
\vspace{0cm}
In theory, the left derivation will be infinite, but a quick inspection of the rules of $\gds{1}{\{D2_1\}}$ (with $Ag=\{1\}$) ensures us that no additional rule application 
 will cause 
the proof to successfully terminate: $\lnot p$ will never be propagated to $z$.  
 The topsequent (left) will give the $\mathsf{DS}_1\{\textbf{D2}\}$-countermodel  for $\oina$ (right), provided that the model is appropriately closed under \textbf{D1} and \textbf{D2}: i.e. $M,w\not\models \oina$ with $W=\{w,v,u,z\}, R_{[1]}=\{(v,z),(z,v)\}, R_{\Box}=\{(v,z),(z,v)\}, R_{\otimes_1}=\{(w,u),(w,v),(u,u),(v,v),(z,v)\}$ and $V(p)=\{w,v,u\}$ (reflexivity is omitted for $R_{[1]}$ and $R_{\Box}$).  
We leave development of terminating proof-search procedures with automated countermodel extraction to future work (cf. \cite{LyoBer19}). 
\end{example}

We close with two remarks: First, 
recall Hintikka's position that 
OiC merely captures the normative disposition that 
 `it \textit{ought} to be that OiC'. An agent-dependent variation of this principle (referred to as $\mathsf{NOiA}$ in Fig.~\ref{fig:oic-relations}) turns out to be a theorem of $\gds{n}{\{D3_i,D4_i\}}$; i.e. $\gds{n}{\{D3_i,D4_i\}}\vdash \Oi (\Oi\phi\rightarrow \Diam [i]\phi)$. 
 Second, we observe that the calculus
  $\gds{n}{\{D5_i\}}$ gives rise to an interesting, yet unaddressed, OiC principle which combines the ideas behind $\oilp$ and $\oina$, namely, $\gds{n}{\{D5_i\}}\vdash \Oi\phi\rightarrow \ODi\Diam[i]\phi$. Loosely, this principle expresses that `ought implies that it is ideally consistent that the agent has the ability to fulfil its duties'. Future research will be directed toward further investigation of the philosophical consequences of our logical taxonomy of deontic $\sti$ logics.
  

\section{Conclusion}\label{sect:conclusion}

In this work, we analyzed, formalized, and compared ten distinct readings of Ought-implies-Can as taken from the philosophical literature. 
We modified the deontic $\sti$ setting to accommodate this variety of OiC principles.  
 Sound and complete deontic $\sti$ calculi were provided of  
 which the aforementioned OiC principles were shown to be theorems. 
 We used these calculi to determine the logical interdependencies between these 
 principles, resulting in a logical taxonomy of Ought-implies-Can according  to each principle's respective strength. In particular, we proposed an endorsement principle describing which OiC readings commit one to other readings logically entailed by the former.  
 
Future work will be twofold: First, from a technical perspective, we aim 
   to provide decision algorithms 
based on the deontic $\sti$ calculi $\gds{n}{X}$, 
 following the work in
  \cite{LyoBer19}. Thus, we will leverage our calculi for the desired automation of normative reasoning within 
  $\sti$.  Furthermore, we aim to logically capture the deliberative OiC principles, bypassing the use of defined deliberative 
 operators.   
 Second, from a more philosophical perspective, future work will be directed toward the identification and analysis of further OiC principles derived from our logical taxonomy of deontic $\sti$ logics. 
 



\Appendix
\section{Soundness and Completeness Proofs}\label{appendix}

\begin{theorem}[Soundness]\label{thm:sound-dstitn} If a sequent $\Lambda$ is derivable in $\mathsf{G3DS}_{n}\mathsf{X}$, then it is valid relative to $\mathsf{DS}_{n}\mathsf{X}$. 
\end{theorem}

\begin{proof} It suffices to show that $\id$ is valid and each rule of $\mathsf{G3DS}_{n}\mathsf{X}$ preserves validity relative to $\mathsf{DS}_{n}\mathsf{X}$. With the exception of $\dfive = \langle \dfivei, \dfiveii \rangle$, all cases are relatively straightforward (cf.~\cite{BerLyo19a,LyoBer19}). The $\dfive$ case follows from the general soundness result for systems of rules presented in~\cite{Neg14}.
\end{proof}

\begin{lemma}\label{lem:compl-dstitn} For any sequent $\Lambda$, either $\Lambda$ is provable in $\mathsf{G3DS}_{n}\mathsf{X}$, or there exists a $\mathsf{DS}_{n}\mathsf{X}$-model $M$ with $I$ such that $M, I \not\models \Lambda$.
\end{lemma}

\begin{proof}
For the proof we expand on the methods employed in \cite{Neg09}. In brief, we first (1) define a reduction-tree $\RT$ for an arbitrary sequent $\Lambda = \rel \vdash \Gamma$. Either $\RT$ terminates and represents a proof in $\mathsf{G3DS}_{n}\mathsf{X}$, implying the provability of $\Lambda$, or it does not terminate. In the latter case the tree will be infinite and, using K\"onig's Lemma, we therefore know that (at least) one of $\RT$'s branches is infinite. We use this infinite branch to show that (2) a $\mathsf{DS}_{n}\mathsf{X}$-model $M$ can be constructed with an interpretation $I$ such that $M, I \not\models \Lambda$.

(1) The inductive construction of $\RT$ consists of phases, each phase having two cases: (i) if every topmost sequent of every branch of $\RT$ is an initial sequent $\id$ the construction terminates. (ii) If not, then for those open branches, the construction proceeds and we continue applying---when possible---the rules of the calculus in a roundabout fashion. (NB. If no rule can be applied to a top sequent, yet it is not an initial sequent, then we copy the top sequent indefinitely.) We show how the $(\agdia)$ and $\dfive$ rules are applied (bottom-up) below; all remaining cases are similar or simple (cf.~\cite{BerLyo19a,Neg09}).

We first consider the $(\agdia)$ case, and suppose that $m$ top sequents $\Lambda_j = \rel_{j} \vdash \Gamma_{j}$ (with $1 \leq j \leq m$) are open in $\RT$ (i.e. no $\Lambda_j$ is an instance of the $\id$ rule). Let $x_1 : \langle i \rangle \phi_{1},...,x_{k_{j}} : \langle i \rangle \phi_{k_{j}}$ be all labelled formulae in $\Lambda_{j}$ prefixed with a $\langle i \rangle$ modality. Moreover, let $y_{l,1}, \ldots, y_{l,r_{l}} \in Lab(\Lambda_{j})$ s.t. $x_{l} \sim^{\rel_{j}}_{i} y_{l,s}$ (for $1 \leq l \leq k_{j}$ and $1 \leq s \leq r_{l}$). We add $\Lambda_{j+1} = \rel_{j} \vdash y_{1,1} : \phi_{1}, \ldots, y_{1,r_{1}} : \phi_{1}, \ldots, y_{k_{j},1} : \phi_{k_{j}}, \ldots, y_{k_{j},r_{k_{j}}} : \phi_{k_{j}}, \Gamma_{j}$ on top of $\Lambda_{j}$. We apply this procedure for all $i\in Ag$.

For the $\dfive$ case, assume that $m$ top sequents $\Lambda_j = \rel_{j} \vdash \Gamma_{j}$ (with $1 \leq j \leq m$) are still open in $\RT$. First, for all $x_1,...,x_{k_{j}} \in Lab(\Lambda_j)$, we set $\rel_{j+1} := \opt_{\Oi} x_1y_1,...,\opt_{\Oi} x_{k_{j}}y_{k_{j}}, \rel_{j}$, set $\Gamma_{j+1} := \Gamma_{j}$, and add $\Lambda_{j+1} = \rel_{j+1} \vdash \Gamma_{j+1}$ on top of $\Lambda_j$, where $y_1,...,y_{k_{j}}$ are fresh. (NB. This corresponds to applications of $\dfivei$.) Second, for all $z_{1}', \ldots, z_{l_{r}}' \in Lab(\Lambda_{j+1})$ such that $z_{r} \sim^{\rel_{j+1}}_{i} z_{1}'$, $\ldots$, $z_{r} \sim^{\rel_{j+1}}_{i} z_{l_{r}}'$ and $\opt_{\Oi}x_{r}'z_{r}$ was introduced by an application of $\dfivei$ at any stage $s \leq j$ (with $1 \leq r \leq h$), we add $\Lambda_{j+2} = \opt_{\Oi}x_{1}'z_{1}',...,\opt_{\Oi}x_{1}',z_{l_{1}}', \ldots, \opt_{\Oi}x_{h}'z_{1}',...,\opt_{\Oi}x_{h}',z_{l_{h}}', \rel_{j+1} \vdash \Gamma_{j+1}$ on top of $\Lambda_{j+1}$. We apply this procedure for all agents $i\in Ag$.

(2) If the construction of the $\RT$ for $\Lambda$ terminates, we know that the topmost sequents of all branches are initial sequents and hence $\RT$ corresponds to a proof. If $\RT$ does not terminate, the tree is infinite and, with K\"onig's Lemma, we obtain an infinite branch from which we can construct a $\mathsf{DS}_{n}\mathsf{X}$ counter-model for $\Lambda$. Let $\rel_{0} \vdash \Gamma_{0},...,\rel_{j} \vdash \Gamma_{j},...$ be the sequence of sequents from the infinite branch, such that, (i) $\Lambda = \rel_{0} \vdash \Gamma_{0}$ and (ii) $\Lambda^+ = \rel^{+} \vdash \Gamma^{+}$, where $\rel^{+} = \bigcup_{j \geq 0} \rel_j$ and $\Gamma^{+} = \bigcup_{j \geq 0} \Gamma_{j}$. 

We construct a model $M^+ = \langle W, R_{\Box}, \{R_{[i]} | i \in Ag\}, \{R_{\Oi} | i \in Ag\}, V\rangle$ as follows: 
$W := Lab(\Lambda^+)$; $\RB := \{(x,y) \ | \ x \sim^{\rel^{+}}_{\Diam} y\}$;  $\RI := \{(x,y) \ | \ x \sim^{\rel^{+}}_{i} y\}$ (for all $i\in Ag$); $\RO := \{ (x,y) \ | \ \opt_{\Oi} xy \in \rel^+ \}$ (for all $i\in Ag$); last, $x \in V(p)$ iff $x: \overline{p} \in \Gamma^+$. 
It is straightforward to show that $M^{+}$ is a $\mathsf{DS}_{n}\mathsf{X}$-model. We show that $M^{+}$ satisfies \textbf{C2} and \textbf{D5} (assuming that \textbf{D5} $\in \mathsf{X}$). The cases for all other conditions \textbf{C1}, \textbf{C3}, \textbf{D1}, and those in $\mathsf{X}$ are similar or simple.

To show that $M^{+}$ satisfies \textbf{C2} we need to show (i) $\R_{[i]} \subseteq \R_{\Box}$, and (ii) $\RI$ is an equivalence relation. To show (i), assume that $(x,y) \in \RI$. This implies that $x \sim^{\rel^{+}}_{i} y$ holds, which further implies that $x \sim^{\rel^{+}}_{\Diamond} y$ holds by Def.~\ref{def:stit-path} and~\ref{def:sett-path}. Therefore, by the definition of $\RB$ in $M^{+}$ above, $(x,y) \in \RB$. To see that $\RI$ is an equivalence relation, it suffices to observe that the relation is defined relative to $\sim^{\rel^{+}}_{i}$, which is an equivalence relation.


To prove that $M^{+}$ satisfies \textbf{D5}, we assume $x\in W$. By the definition of $\RT$, we know that there exists a $\Lambda_j$ in the infinite branch such that $x \in Lab(\Lambda_j)$. Since the branch is infinite and rules are applied in a roundabout fashion we know that at some point $k>j$ the $\dfive$ step of the $\RT$ procedure must have been applied (and so, $\dfivei$ must have been applied). Hence, $\opt_{\Oi} xy \in \rel_{k+1}$ for $\Lambda_{k+1} = \rel_{k+1} \vdash \Gamma_{k+1}$ with $y$ fresh, implying that $(x,y)\in \RO$. We aim to show that for all $z\in W$, if $(y,z)\in \RI$, then $(x,z)\in\RO$. Take an arbitrary $z\in W$ for which $(y,z)\in\RI$. By the assumption that $(y,z)\in\RI$ and by the definition of $\RT$, we know that at some point $m \geq k+1$ that the $\dfive$ step of the $\RT$ procedure must have been applied (and so, $\dfiveii$ must have been applied) with $y \sim^{\rel_{m}}_{i} z$ for $\Lambda_{m} = \rel_{m} \vdash \Gamma_{m}$. Hence, $\RO xz \in \rel_{m+1}$ in $\Lambda_{m+1} = \rel_{m+1} \vdash \Gamma_{m+1}$, implying that $(x,z) \in \RO$.

Let $I: Lab \mapsto W$ be the identity function (we may assume w.l.o.g. that $Lab = W$). By construction, $M^+$ satisfies each relational atom occurring in $\rel^+$ with $I$, meaning that $M^+$ satisfies each relational atom in $\rel$ with $I$ (recall $\Lambda = \rel \vdash \Gamma$). It can be shown by induction on the complexity of $\phi$ that for any $x: \phi \in \Gamma^+$, $M^{+}, I(x) \not\models \phi$. Consequently, since $\Gamma \subseteq \Gamma^{+}$, 
$M^{+},I \not\models \Lambda$. 
\end{proof}

\begin{theorem}[Completeness]\label{thm:compl-dstitn} If a sequent $\Lambda$ is valid relative to $\mathsf{DS}_{n}\mathsf{X}$, then it is derivable in $\mathsf{G3DS}_{n}\mathsf{X}$.
\end{theorem}

\begin{proof}
Follows directly from \ref{lem:compl-dstitn}.
\end{proof}




%
%

\bibliographystyle{deon16}
\bibliography{biblio}

\end{document}